\documentclass[11pt,final,onecolumn]{IEEEtran}
\usepackage{amssymb}
\usepackage{amstext}
\usepackage{amsmath}
\usepackage{epsfig}
\usepackage{epic,eepic}

\newtheorem{theorem}{Theorem}[section]
\newtheorem{lemma}[theorem]{Lemma}

\newcommand{\qed}{\nobreak \ifvmode \relax \else
     \ifdim\lastskip<1.5em \hskip-\lastskip
      \hskip1.5em plus0em minus0.5em \fi \nobreak
      \vrule height0.75em width0.5em depth0.25em\fi}

\begin{document}
\title{Generalized Degrees of Freedom of the Symmetric Gaussian $K$ User Interference Channel}
\author{
\authorblockN{Syed A. Jafar\\}
\authorblockA{Electrical Engineering and Computer Science\\
University of California Irvine, \\
Irvine, California, 92697, USA\\
Email: syed@uci.edu\\ \vspace{1cm}}
\and
\authorblockN{Sriram Vishwanath\\}
\authorblockA{Department of Electrical Engineering \\
University of Texas Austin\\
Austin, Texas, USA\\
Email: sriram@ece.utexas.edu\\ \vspace{-1cm}
}}

\maketitle
\IEEEpeerreviewmaketitle

\begin{abstract}
We characterize the generalized degrees of freedom of the $K$ user symmetric Gaussian interference channel where all desired links have the same signal-to-noise ratio (SNR) and all undesired links carrying interference have the same interference-to-noise ratio, $\mbox{INR}=\mbox{SNR}^\alpha$. We find that the number of generalized degrees of freedom per user, $d(\alpha)$, does not depend on the number of users, so that the characterization is identical to the $2$ user interference channel with the exception of a singularity at $\alpha=1$ where $d(1)=\frac{1}{K}$. The achievable schemes use multilevel coding with a nested lattice structure that opens the possibility that the sum of interfering signals can be decoded at a receiver even though the messages carried by the interfering signals are not decodable.
\end{abstract}
\newpage
\section{Introduction}
The capacity of the two user Gaussian interference channel has recently been characterized upto an accuracy of one bit per channel use \cite{Etkin_Tse_Wang}. A closely related notion, also introduced in \cite{Etkin_Tse_Wang}, is that of the generalized degrees of freedom. The generalized degrees of freedom provide a clear mapping of various interference management approaches to the SNR and INR regimes where they are optimal. For designing wireless networks, the natural question is how the new insights provided in \cite{Etkin_Tse_Wang} generalize to communication scenarios with more than $2$ users. Unfortunately, for interference networks with more than $2$ users, even the conventional degrees of freedom are not known in many cases, e.g. when channel coefficients take arbitrary but constant values. Degrees of freedom are known when channel coefficients are time-varying or frequency selective and in some cases if the nodes are equipped with multiple antennas \cite{Cadambe_Jafar_int, Cadambe_Jafar_X}. The key insight from the degree of freedom characterizations is the possibility of interference alignment in wireless networks by coding across the dimension (time, frequency, space) along which the channel coefficients vary. For channels with single antenna nodes and constant coefficients (i.e., channels that do not vary in time, frequency or space dimensions), deterministic chanel models proposed in \cite{Avestimehr_Diggavi_Tse} identify a new dimension - the signal \emph{level} along which the relativity of alignment can be exploited to achieve interference alignment \cite{Bresler_Parekh_Tse, Cadambe_Jafar_Shamai}. The use of deterministic channel models has led to approximate capacity characterizations for the many-to-one and one-to-many Gaussian interference channels \cite{Bresler_Parekh_Tse}. Deterministic channel models are also used in \cite{Cadambe_Jafar_Shamai} to characterize the conventional degrees of freedom for some special instances of Gaussian interference channels with constant channel coefficients. However, in the richer context of generalized degrees of freedom, there are no known characterizations for more than $2$ users. Thus, it is not clear if interference regimes elegantly identified by \cite{Etkin_Tse_Wang} generalize to interference channels with more than $2$ users. In this work, our goal is to advance the generalized degrees of freedom perspective to $K$ user \emph{symmetric} Gaussian interference channels with constant channel coefficients. We start with the channel model.

\section{The Symmetric Gaussian Interference Channel}\label{sec:systemmodel}
We consider the $K$ user interference network described by the input-output equations:
\begin{eqnarray*}
Y^{[k]}(t)=X^{[k]}(t)+\sqrt{\frac{\mbox{INR}}{\mbox{SNR}}}\sum_{j=1,j\neq k}^KX^{[j]}(t)+Z^{[k]}(t)\label{eq:inout}
\end{eqnarray*}
where at discrete time index $t$, $Y^{[k]}(t)$ and $Z^{[k]}(t)$ are the channel output and additive white Gaussian noise (AWGN), respectively, at the $k^{th}$ receiver and $X^{[j]}(t)$ is the channel input symbol at the $j^{th}$ transmitter, $\forall j,k\in\mathcal{K}\triangleq\{1,2,\cdots,K\}$. All symbols are real and the channel coefficients are fixed. The time index $t$ is suppressed henceforth for compact notation.

The AWGN is normalized to have zero mean and unit variance and the input power constraint is given by
\begin{eqnarray*}
\mbox{E}\left[(X^{[k]})^2\right]\leq \mbox{SNR}, ~~~\forall k\in\mathcal{K}.\label{eq:powerconstraint}
\end{eqnarray*}
The interference-to-noise-ratio (INR) is defined through the parameter $\alpha$ as:
\begin{eqnarray}
\frac{\log(\mbox{INR})}{\log(\mbox{SNR})}&=&\alpha\\
\Rightarrow \mbox{INR}&=& \mbox{SNR}^\alpha
\end{eqnarray}
Thus, as in \cite{Etkin_Tse_Wang}, the channel is parameterized by $\alpha$. We define the generalized degrees of freedom \emph{per user} as:
\begin{eqnarray}
d(\alpha)&=&\frac{1}{K}\limsup_{\mbox{SNR}\rightarrow\infty}\frac{C_\Sigma(\mbox{SNR},\alpha)}{\frac{1}{2}\log(\mbox{SNR})}
\end{eqnarray}
where $C_\Sigma(\mbox{SNR},\alpha)$ is the \emph{sum}-capacity of the $K$ user interference network defined above. Note that we use $\limsup$ to ensure that $d(\alpha)$ is always defined, with the understanding that $\limsup$ and $\lim$ are the same if the ordinary limit exists. Finally, the half in the denominator is because we are dealing with real signals only.

\section{Generalized Degrees of Freedom (GDOF)}
The generalized degrees of freedom for the $K$ user interference channel described above, are presented in the following theorem.
\begin{theorem}\label{theorem:main}
\begin{eqnarray}
d(\alpha)\leq\left\{
\begin{array}{lrl}
1-\alpha,&0\leq\alpha\leq\frac{1}{2}&\mbox{(noisy interference)}\\
\alpha,&\frac{1}{2}\leq\alpha\leq\frac{2}{3}&\mbox{(weak interference)}\\
1-\frac{\alpha}{2},&\frac{2}{3}\leq\alpha< 1&\mbox{(moderately weak interference)}\\
\frac{1}{K},&\alpha=1&\\
\frac{\alpha}{2},&1<\alpha\leq 2&\mbox{(strong interference)}\\
1,&\alpha\geq 2&\mbox{(very strong interference)}
\end{array}
\right.
\end{eqnarray}
\end{theorem}
\vspace{0.3cm}
{\it Remark 1:} Note that, with the exception of $\alpha=1$ the generalized degrees of freedom per user do not depend on the number of users, so that the GDOF characterization is identical to that obtained in \cite{Etkin_Tse_Wang} for the $2$ user symmetric Gaussian interference channel. Also note that the GDOF characterization shows a singularity at $\alpha=1$ when the number of users $K>2$.
\begin{figure}
\centerline{\input{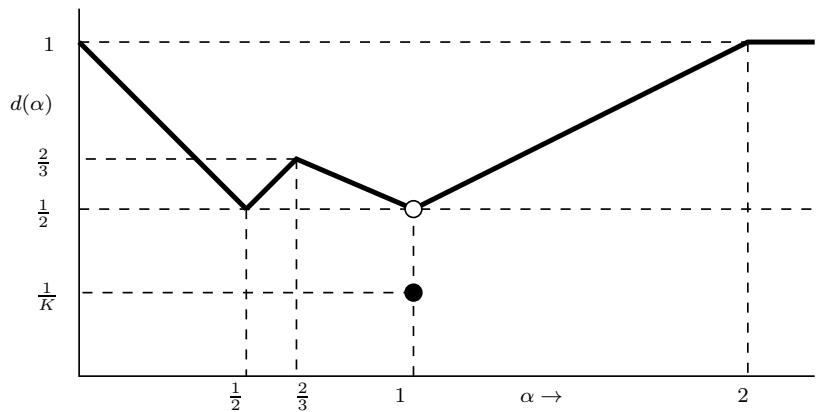}}
\caption{Generalized degrees of freedom for the $K$ user symmetric Gaussian interference channel}
\end{figure}
\subsection{Proof of Outerbound}
For the case $\alpha=1$, all receivers see statistically equivalent signals, which implies that all messages can be decoded by all receivers and the sum capacity is the multiple access capacity from the $K$ transmitters to, say, receiver 1. Mathematically, $C_\Sigma(\mbox{SNR},1)=\frac{1}{2}\log(1+K\mbox{SNR})$ which implies that $d(1)=\frac{1}{K}$. Note that achievability of $1/K$ degrees of freedom per user is trivial.

For $\alpha\neq 1$ the outerbound is straightforward as well. Eliminating all but two users, say users 1 and 2, the generalized degrees of freedom characterization for the remaining two users is given by the result of \cite{Etkin_Tse_Wang} which is the same as the outerbound above. Adding back other users does not help these two users' rates, so the result of \cite{Etkin_Tse_Wang} becomes an outerbound on the degrees of freedom achieved by users 1 and 2 in the $K$ user interference channel. The same argument applies for any subset of two users. Combining these outerbounds produces the outerbound of Theorem \ref{theorem:main}. Note that the outerbound is tight for $K=2$, as established in \cite{Etkin_Tse_Wang}.

Next we establish tight innerbounds for each regime. For this purpose we break the result of Theorem \ref{theorem:main} into several lemmas, each of which is covered in a separate section.
\section{Achievability of Generalized Degrees of Freedom}
\subsection{Preliminary Setup}
In this paper, all the achievable schemes for the generalized degrees of freedom are obtained by first finding the optimal scheme for the deterministic channel model and then translating it to the Gaussian channel, in a similar manner as \cite{Cadambe_Jafar_Shamai}. To avoid repetition, we state the common aspects of the achievable schemes before addressing the individual scenarios. 

For a given $\alpha\neq 1$, consider the sequence of SNR values (indexed by $M$) such that:
\begin{eqnarray}
\mbox{SNR}&=&Q^\frac{2M}{\left|\alpha-1\right|}\label{eq:SNRQ}
\end{eqnarray}
where $Q\gg K$ is a large but fixed positive integer value and $M$ is a positive integer whose value grows to infinity producing a sequence of SNR values that approach infinity as well. Using (\ref{eq:SNRQ}), we can rewrite the input-output equation (\ref{eq:inout}) as:
\begin{eqnarray*}
Y^{[k]}=X^{[k]}+Q^{\mbox{sgn}(\alpha-1)M}\sum_{j=1,j\neq k}^KX^{[j]}+Z^{[k]}
\end{eqnarray*}
where $\mbox{sgn}(x)=1$ if $x>0$ and $-1$ if $x<0$.

To compute lowerbounds on the generalized degrees of freedom we will express both the symmetric rate $R_{\mbox{sym}}$, i.e. a rate achieved simultaneously by each user, and the SNR as functions of $M, Q, \alpha$. We use $\log_Q(\cdot)$ for both rate and SNR. Thus,
\begin{eqnarray}
d(\alpha)&=&\frac{1}{K}\limsup_{\mbox{SNR}\rightarrow\infty}\frac{C_\Sigma(\mbox{SNR},\alpha)}{\frac{1}{2}\log(\mbox{SNR})}\geq\limsup_{\mbox{M}\rightarrow\infty}\frac{R_{\mbox{sym}}(M,Q,\alpha)}{M/|\alpha-1|}\label{eq:gdofM}
\end{eqnarray}
where we used (\ref{eq:SNRQ}) to make the substitution $\frac{1}{2}\log_Q(\mbox{SNR})=M/|\alpha-1|$.

We will occasionally represent \emph{positive} real signals in base-$Q$ notation using $Q-$ary digits $0,1,2,\cdots, Q-1$, which we denote as ``qits''. To avoid confusion we will mark the $Q$-ary representation as $[\cdot]_Q$. For example, when we write:
\begin{eqnarray*}
A=\left[\cdots A_3A_2A_1A_0.A_{-1}A_{-2}A_{-3}\cdots\right]_Q
\end{eqnarray*}
then it is implied that $A_i$  are integers with values between $0$ and $Q-1$, or the qits in the Q-ary representation of the real number $A$. Equivalently,
\begin{eqnarray*}
A=\sum_{i=-\infty}^{\infty} A_iQ^i, ~~A_i\in\{0,1,\cdots, Q-1\}
\end{eqnarray*}

\subsubsection{Transmit Scheme}

We impose some structure on the $Q-$ary representation of the transmit signal $X^{[k]}$, which is expressed as:
\begin{eqnarray}
X^{[k]}=\left[X^{[k]}_{N-1}\cdots X^{[k]}_3X^{[k]}_2X^{[k]}_1X^{[k]}_0.X^{[k]}_{-1}X^{[k]}_{-2}X^{[k]}_{-3}\cdots X^{[k]}_{-L}\right]_Q
\end{eqnarray}
i.e., the most significant (non-zero) qit corresponds to the coefficient of $Q^{N-1}$ and the least significant (non-zero) qit corresponds to the coefficient of $Q^{-L}$ for some $N,L$. The same structure applies to all users,  $k\in\mathcal{K}$. Additional structure may also be imposed on the transmitted signals for different interference regimes considered in subsequent sections. In each case, the values of the qits are restricted to ensure that (in the absence of noise) addition of interfering signals does not produce carry overs. For example this could be accomplished by the restriction that $X^{[k]}_i\in\{0,1,\cdots,\lfloor\frac{Q-1}{K}\rfloor\}$.

Note that $0\leq X^{[k]} \leq Q^N$. Thus, the transmitted power $\mbox{E}\left[(X^{[k]})^2\right]<Q^{2N}$. The power constraint (\ref{eq:powerconstraint}) is satisfied if $Q^{2N}\leq SNR$. This is guaranteed by the choice:
\begin{eqnarray}
N =\left\lfloor \frac{M}{|\alpha-1|}\right\rfloor \label{eq:chooseN}
\end{eqnarray}
Each qit is coded in time independently of other qits in the manner of multi-level coding. Thus, over $T$ channel uses, the codeword for qit $i$ from transmitter $k$ is $X^{[k]}_i(1), X^{[k]}_i(2),\cdots,X^{[k]}_i(T)$.

\subsubsection{Receive Scheme}
Each receiver $k\in\mathcal{K}$ takes the magnitude of the received signal, reduces it modulo $Q^{m}$, discards the value below the decimal point and expresses the result  in $Q-$ary representation as:
\begin{eqnarray}
\overline{Y}^{[k]}=\left\lfloor|Y^{[k]}|\mod Q^{m}\right\rfloor&=& \sum_{i=0}^{m-1} \overline{Y}^{[k]}_iQ^i, ~~\overline{Y}^{[k]}_i\in\{0,1,\cdots, Q-1\}
\end{eqnarray}
Here, $m=\max(N,N+\mbox{sgn}(\alpha-1)M)$ is the maximum number of qits seen at the receiver above the decimal place (noise floor) that are contributed by the transmitted signals.  The $k^{th}$ receiver views the $i^{th}$ qit as a separate channel with input $X^{[k]}_i$ and output $\overline{Y}^{[k]}_i$.

Let us also define a noise-free version of the received signal,
\begin{eqnarray*}
\hat{Y}^{[k]}=X^{[k]}+Q^{\mbox{sgn}(\alpha-1)M}\sum_{j=1,j\neq k}^KX^{[j]}
\end{eqnarray*}
which is essentially the received signal in the deterministic channel model of \cite{Avestimehr_Diggavi_Tse}. 

Let $P^e_i$ be the probability that the $i^{th}$ qit of $\hat{Y}^{[k]}$ is not the same as the $i^{th}$ qit of $\overline{Y}^{[k]}$, $0\leq i\leq m-1$. The key idea is that due to the structure imposed on the transmitted qits and the finite variance of noise, this probability $P^e_i$ is a motonically decreasing function of $i$ and approaches zero as $i$ becomes large. Thus, as $M, N$ and $i$ become large, the limiting rate achieved with noise-free channel output $\hat{Y}^{[k]}_i$ is identical to the rate achieved with output $\overline{Y}^{[k]}_i$ by coding across channel uses over each qit separately. We provide a detailed exposition of this argument only for the very strong interference case.  The argument applies to all achievable schemes in this paper in a straightforward fashion. For other applications of the same argument see \cite{Cadambe_Jafar_Shamai}. 

We start with the very strong interference regime.
\subsection{Very Strong Interference}\label{sec:verystrong}
The following theorem establishes the generalized degrees of freedom for the very strong interference case.
\begin{lemma}
For the $K$ user symmetric Gaussian channel defined in Section \ref{sec:systemmodel}, the generalized degrees of freedom $d(\alpha)=1$ for $\alpha\geq 2$. 
\end{lemma}
\begin{proof}
The transmitted symbol of user $j$ is constructed as follows:
\begin{eqnarray}
X^{[j]}&=&\left[X^{[j]}_{N-1}X^{[j]}_{N-2}X^{[j]}_{N-3}\cdots X^{[j]}_1 X^{[j]}_0.0\right]_Q
\end{eqnarray}
The same construction is used for all transmitters. We restrict the values of the qits as:
\begin{eqnarray}
X^{[j]}_i\in\{1,2,\cdots,Q-2\}, \forall j\in\{1,2,\cdots,K\}, i\in\{0,1,\cdots,N-1\}
\end{eqnarray}
Note that the qit values $0$ and $Q-1$ are not used. While not necessary, this is done to simplify the argument that the impact of the noise is limited and it does not propagate indefinitely through carry overs into higher signal levels.
The value of $N$ is chosen to satisfy the power constraint
\begin{eqnarray}
N =\left\lfloor \frac{M}{|\alpha-1|}\right\rfloor 
\end{eqnarray}
as in (\ref{eq:chooseN}). Since $\alpha\geq 2$ note that this implies that $M\geq N$.

\begin{figure}
\centerline{\input{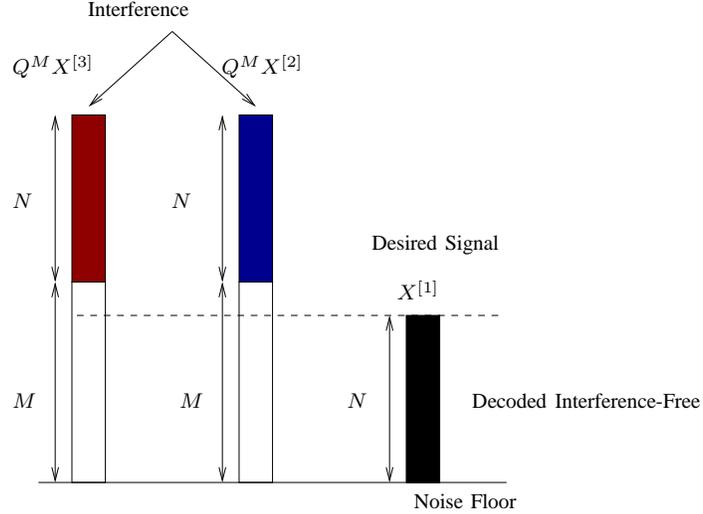}}
\caption{Signal Levels at Receiver 1 for Very Strong Interference Case}
\end{figure}

The key idea is that  multiplication by $Q^{M}$ shifts the decimal point in the Q-ary representation $(X^{[j]})_Q$ by $M$ places to the right. Thus, an interfering signal reaches a receiver shifted up by $M$ qits. Since $M\geq N$ all the interference is easily eliminated by a simple modulo $Q^N$ operation at the receiver. To account for noise, let us evaluate the probability $P^e_i$ for $0\leq i\leq N$. Note that:
\begin{eqnarray}
1-P^e_i=\mbox{Prob}(\hat{Y}^{[k]}_i=\overline{Y}^{[k]}_i)\geq \mbox{Prob}\left(\left|Z^{[k]}\right|\leq Q^{i-1}\right)
\end{eqnarray}
because regardless of whether it is positive or negative, any additive noise with magnitude no more than $Q^{i-1}$ does not affect the coefficient of $Q^i$. Thus, $P^e_i$ is monotonically decreasing in $i$ and goes to $0$ as $i$ becomes large. Now consider the channel with input $X^{[k]}_i$ and noise-free output $\hat{Y}^{[k]}_i$ as defined above. Because of the shift of the interfering qits out of the desired space we have:
\begin{eqnarray}
\hat{Y}^{[k]}_i=X^{[k]}_i, 0\leq i\leq N-1.
\end{eqnarray}
Thus, this channel has channel capacity $\hat{C}_i=\log_Q(Q-2)$ qits per channel use.

The noisy channel has input $X^{[k]}_i$ and output $\overline{Y}^{[k]}_i$. However because $P^e_i$ approaches zero as $i\rightarrow\infty$ the capacity of the noisy channel can be expressed as:
\begin{eqnarray}
\lim_{i\leq N-1, i\rightarrow\infty}\overline{C}_i &=&\hat{C}_i\\
\Rightarrow \sum_{i=0}^{N-1}\overline{C}_i&=& N\log_Q(Q-2)+o(N)
\end{eqnarray}
Thus, we express the symmetric achievable rate for the proposed scheme as:
\begin{eqnarray}
R_{\mbox{sym}}&=&\frac{M}{\alpha-1}\left(\log_Q(Q-2)\right)+o(M)\nonumber
\end{eqnarray}
Substituting into (\ref{eq:gdofM}) we have
\begin{eqnarray}
d(\alpha)&\geq& \limsup_{M\rightarrow\infty}\frac{\frac{M}{\alpha-1}\left(\log_Q(Q-2)\right)+o(M)}{M/(\alpha-1)}\\
&=& \log_Q(Q-2)
\end{eqnarray}
Observing that $Q$ can be chosen to be arbitrarily large and comparing with the outerbound  for $\alpha\geq 2$, we have 
\begin{eqnarray}
d(\alpha)&=& 1, ~~~\forall \alpha\geq 2.
\end{eqnarray}

\end{proof}
\subsection{Strong Interference}
The following theorem establishes the generalized degrees of freedom for the strong interference case.
\begin{lemma}
For the $K$ user symmetric Gaussian channel, $d(\alpha)=\alpha/2$ for $1< \alpha\leq 2$.
\end{lemma}
\begin{proof}
The transmitted symbol of user $j$ is constructed as follows:
\begin{eqnarray}
X^{[j]}&=&\left[\underbrace{X^{[j]}_{2N-M-1}\cdots X^{[j]}_{N}}_{N-M \mbox{ ``copy'' qits}}\underbrace{X^{[j]}_{N-1}\cdots  X^{[j]}_0}_{N \mbox{ qits}}.0\right]_Q
\end{eqnarray}
The same construction is used for all transmitters. We restrict the values of the qits as:
\begin{eqnarray}
X^{[k]}_i\in\{1,2,\cdots,\lfloor\frac{Q-1}{K}\rfloor-1\}
\end{eqnarray}
This ensures that no carryovers are produced by the addition of the interfering signals in the absence of noise and that $P^e_i$ approaches zero as $i$ becomes large. Further, we require that:
\begin{eqnarray}
X^{[j]}_{2N-M-i}=X^{[j]}_{i-1}
\end{eqnarray} 
for $i\in\{1,2,\cdots,N-M\}$. Thus the most significant $N-M$ qits are simply copies of the least significant $N-M$ qits in opposite order.

With this construction the transmit power  can be expressed in the form of the following outerbound:
\begin{eqnarray}
\mbox{E}\left[(X^{[j]})^2\right]&\leq&\left(Q^{2N-M}\right)^2
\end{eqnarray}

The power constraint is satisfied if $2N-M\leq\frac{M}{\alpha-1}$. So we choose:
\begin{eqnarray}
 N&=&\left\lfloor\frac{M\alpha}{2(\alpha-1)}\right\rfloor\geq M
\end{eqnarray}
\begin{figure}
\centerline{\input{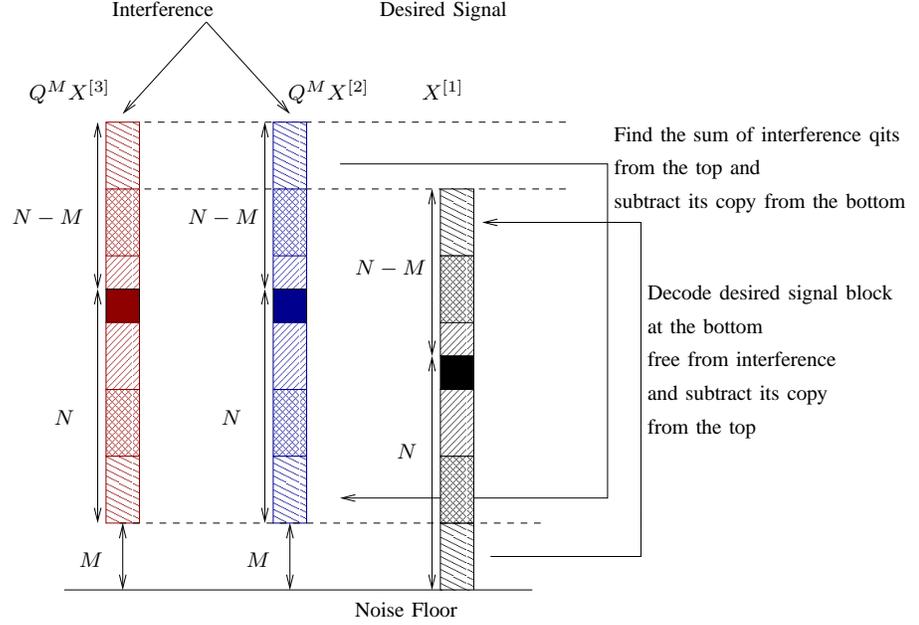}}
\caption{Signal Levels at Receiver 1 for Strong Interference Case}\label{fig:strong}
\end{figure}
The decoding takes place stepwise in blocks of size $M$ qits, as illustrated in Figure \ref{fig:strong}. In each step the least significant block of $M$ qits of desired signal is decoded free from interference and its copy present in the the most significant $M$-qit block of the desired signal, is subtracted out. At the same time, the most significant $M$-qit block of the \emph{sum} of the interferers' signals is used to cancel out the least significant $M$-qit block of the sum of interfering qits. Thus, each step produces a new $M$-qit block of desired signal free from interference and a new $M$-qit block of interfering signals free from the desired signal. The procedure is continued until all $N$ qits of desired signal are decoded.

Note that the interference qits are not decoded, but rather the sum of the interference codeword symbols and noise is subtracted from its copy. Thus, the cancellation of interference introduces noise in the desired signal. However, the probability $P^e_i$ once again approaches zero as the block size $M$ increases and the same arguments (as explained in the very strong interference case) apply for characterizing the rate of desired signal within $o(M)$ for each $M$-qit block. Using this achievable scheme, the symmetric rate achieved simultaneously by all users is:
\begin{eqnarray}
R_{\mbox{sym}}&=&\frac{M\alpha}{2(\alpha-1)}\left(\log_Q(\lfloor\frac{Q-1}{K}\rfloor -1\right)+o(M)\nonumber
\end{eqnarray}
Thus,
\begin{eqnarray}
d(\alpha)&\geq& \limsup_{M\rightarrow\infty}\frac{\frac{M\alpha}{2(\alpha-1)}\left(\log_Q(\lfloor\frac{Q-1}{K}\rfloor -1)\right)+o(M)}{M/(\alpha-1)}\\
&=& \frac{\alpha}{2}\log_Q(\lfloor\frac{Q-1}{K}\rfloor -1)
\end{eqnarray}
Observing that $Q$ can be chosen to be arbitrarily large while $K$ is fixed, and comparing with the outerbound, we have 
\begin{eqnarray}
d(\alpha)&= \alpha/2, ~~~\forall\alpha\in (1, 2].
\end{eqnarray}

\end{proof}
\subsection{Moderately Weak Interference}
The following theorem establishes the generalized degrees of freedom for the moderately weak interference case.
\begin{lemma}
For the $K$ user symmetric Gaussian channel, $d(\alpha)=1-\frac{\alpha}{2}$ for $\frac{2}{3}\leq \alpha< 1$. 
\end{lemma}
\begin{proof}
The transmitted symbol of user $j$ is constructed as follows:
\begin{eqnarray}
X^{[j]}&=&\left[\underbrace{X^{[j]}_{2N+3M-1}\cdots X^{[j]}_{N+3M}}_{N \mbox{ qits}}\underbrace{X^{[j]}_{3M+N-1}\cdots  X^{[j]}_{2M+N}}_{M \mbox{ qits}}\underbrace{X^{[j]}_{2M+N-1}\cdots X^{[j]}_{2M}}_{N \mbox{ ``copy'' qits}}\underbrace{0\cdots 0}_{M \mbox{ zeros}}\underbrace{X^{[j]}_{M-1}\cdots  X^{[j]}_0}_{M \mbox{ qits}}.0\right]_Q\nonumber
\end{eqnarray}
The same construction is used for all transmitters. All transmitted qits lie between $1$ and $\lfloor\frac{Q-1}{K}\rfloor -1$. Further, we require that
\begin{eqnarray}
X^{[j]}_{2N+3M-i}=X^{[j]}_{2M+i-1}
\end{eqnarray} 
for $i\in\{1,2,\cdots,N\}$. Thus, the most signficant $N$ qits are copied in reverse order.

With this construction the transmit power  can be expressed in the form of the following outerbound:
\begin{eqnarray}
\mbox{E}\left[(X^{[j]})^2\right]&\leq&\left(Q^{2N+3M}\right)^2
\end{eqnarray}

The power constraint is satisfied if $2N+3M\leq\frac{M}{1-\alpha}$, so we choose:
\begin{eqnarray}
 N&=&\left\lfloor\frac{M(3\alpha-2)}{2(1-\alpha)}\right\rfloor
\end{eqnarray}

The key idea for the achievable scheme is illustrated in Figure \ref{fig:modweak}. Decoding takes place in blocks of $M$ qits. Starting from the noise floor, the $M$ least significant qits of the interference appear below the noise floor and are discarded. The $M$ least significant qits of the desired signal coincide with the zero padding qits of the interferers so they can be decoded without interference as well. From this point on the decoding proceeds in similar steps. Each step, the most significant $M$ qits of the desired signal do not see any interference and are decoded interference-free. These qits are subtracted from the copy that appears at the least significant end. Similarly, the sum of the $M$ least significant qits of the interference is computed free from desired signal and this sum is subtracted from the most significant interference qits to produce an interference free desired signal block. The step is repeated until all the desired qits are decoded.
\begin{center}
\begin{figure}
\centerline{\input{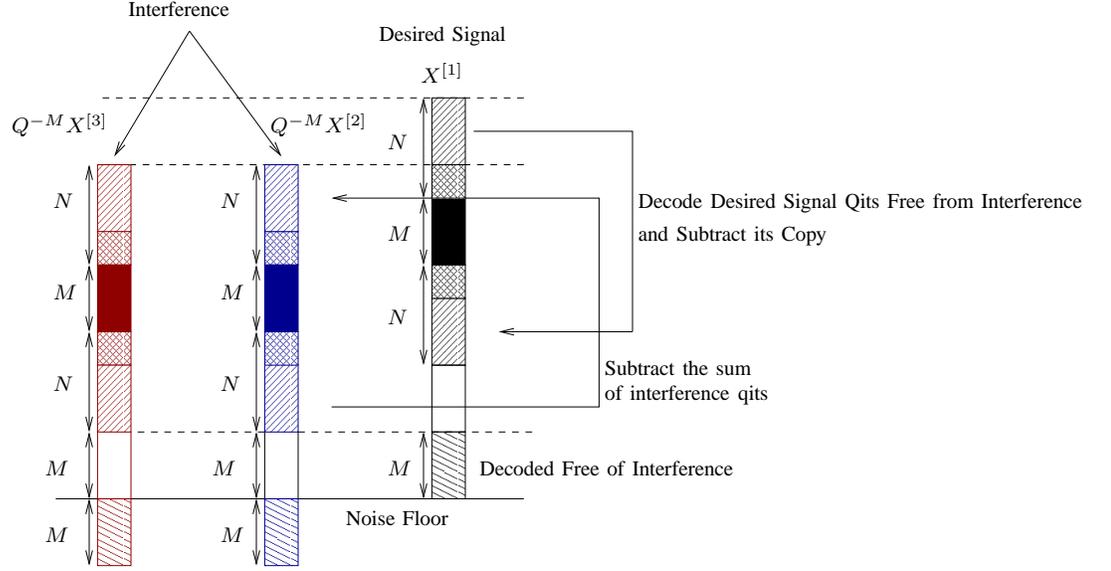}}
\caption{Signal Levels at Receiver 1 for Moderately Weak Interference Case}\label{fig:modweak}
\end{figure}
\end{center}

The symmetric rate achieved with this scheme is:
\begin{eqnarray}
R_{\mbox{sym}}&=&(2M+N)\log_Q\left(\lfloor\frac{Q-1}{K}\rfloor -1\right)+o(M)\nonumber\\
&=&\left(2M+\frac{M(3\alpha-2)}{2(1-\alpha)}\right)\log_Q\left(\lfloor\frac{Q-1}{K}\rfloor -1\right)+o(M)\nonumber\\
&=&\frac{M(1-\alpha/2)}{(1-\alpha)}\log_Q\left(\lfloor\frac{Q-1}{K}\rfloor -1\right)+o(M)\nonumber
\end{eqnarray}
Thus,
\begin{eqnarray}
d(\alpha)&\geq& \limsup_{M\rightarrow\infty}\frac{\frac{M(1-\alpha/2)}{(1-\alpha)}\log_Q\left(\lfloor\frac{Q-1}{K}\rfloor -1\right)+o(M)}{M/(\alpha-1)}\\
&=& \left(1-\frac{\alpha}{2}\right)\log_Q\left(\lfloor\frac{Q-1}{K}\rfloor -1\right)
\end{eqnarray}
Observing that $Q$ can be chosen to be arbitrarily large while $K$ is fixed, and comparing with the outerbound, we have 
\begin{eqnarray}
d(\alpha)&= 1- \alpha/2, ~~~\forall\alpha\in [\frac{2}{3}, 1).
\end{eqnarray}
\end{proof}

\subsection{Weak Interference}
The following theorem establishes the generalized degrees of freedom for the weak interference case.
\begin{lemma}
For the $K$ user symmetric Gaussian channel, $d(\alpha)=\alpha$ for $\frac{1}{2}\leq \alpha\leq \frac{2}{3}$. 
\end{lemma}
\begin{proof}
The transmitted symbol of user $j$ is constructed as follows:
\begin{eqnarray}
X^{[j]}&=&\left[\underbrace{X^{[j]}_{2M+N-1}\cdots X^{[j]}_{2M}}_{N \mbox{ qits}}\underbrace{0\cdots 0}_{M \mbox{ zeros}}\underbrace{X^{[j]}_{M-1}\cdots  X^{[j]}_0}_{M \mbox{ qits}}.0\right]_Q\nonumber
\end{eqnarray}
The same construction is used for all transmitters. All transmitted qits lie between $1$ and $Q-2$.

With this construction the transmit power  can be expressed in the form of the following outerbound:
\begin{eqnarray}
\mbox{E}\left[(X^{[j]})^2\right]&\leq&\left(Q^{N+2M}\right)^2
\end{eqnarray}

The power constraint is satisfied if we choose $ N+2M\leq\frac{M}{(1-\alpha)}$, so we choose
\begin{eqnarray}
N&=&\left\lfloor \frac{M(2\alpha-1)}{1-\alpha}\right\rfloor\leq M
\end{eqnarray}

\begin{figure}
\centerline{\input{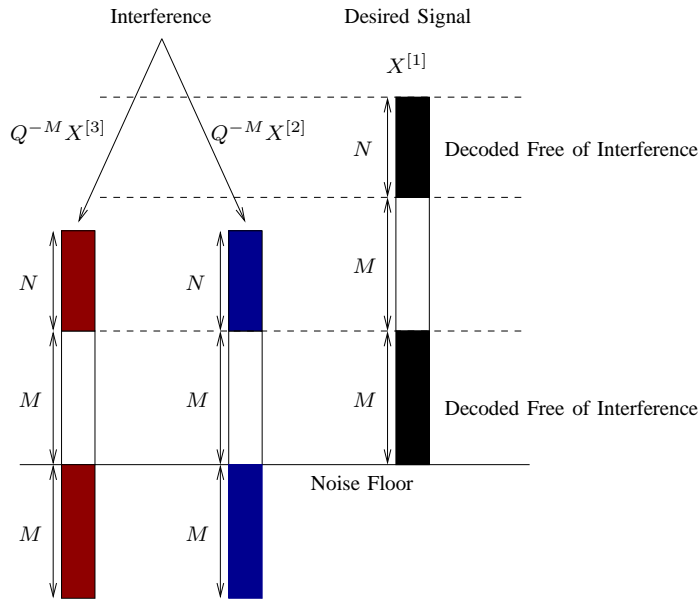}}
\caption{Signal Levels at Receiver 1 for Weak Interference Case}\label{fig:weak}
\end{figure}

The scheme is explained Fig. \ref{fig:weak}. The channel pushes the interference level down by $M$ qits. Thus, the least significant $M$ qits of the desired signal align with the zero padding qits of interference and can be decoded free from interference. Since, $N\leq M$, the most significant $N$ qits are also free from interference.
Thus, the symmetric rate achieved simultaneously by all users is expressed as:
\begin{eqnarray}
R_{\mbox{sym}}&=&\left(\frac{M(2\alpha-1)}{(1-\alpha)}+M\right)\log_Q(Q-2)+o(M)\nonumber
\end{eqnarray}
Thus,
\begin{eqnarray}
d(\alpha)&\geq& \limsup_{M\rightarrow\infty}\frac{\frac{M\alpha}{1-\alpha}\log_Q(Q-2)+o(M)}{M/(\alpha-1)}\\
&=& \alpha\log_Q\left(Q-2\right)
\end{eqnarray}
Observing that $Q$ can be chosen to be arbitrarily large while $K$ is fixed, and comparing with the outerbound, we have 
\begin{eqnarray}
d(\alpha)&=& \alpha, ~~~\forall\alpha\in [\frac{1}{2}, \frac{2}{3}].
\end{eqnarray}
\end{proof}

\subsection{Noisy Interference}
The following theorem establishes the generalized degrees of freedom for the noisy interference case ($0\leq\alpha\leq 1/2$).
\begin{lemma}
For the $K$ user symmetric Gaussian channel, $d(\alpha)=1-\alpha$ for $0\leq \alpha\leq \frac{1}{2}$. An optimal scheme in the degrees of freedom sense is to use Gaussian codebooks and treat interference as noise.
\end{lemma}
\begin{proof}
With Gaussian codebooks and with interference treated as noise, the symmetric achievable rate is:
\begin{eqnarray}
R_{\mbox{sym}}&=&\frac{1}{2}\log\left(1+\frac{\mbox{SNR}}{1+(K-1)\mbox{INR}}\right)\nonumber\\
&=&\frac{1}{2}\log\left(1+\frac{\mbox{SNR}}{1+(K-1)\mbox{SNR}^\alpha}\right)\nonumber\\
&=&\frac{1}{2}(1-\alpha)\log\left(\mbox{SNR}\right)+o(\log(\mbox{SNR}))\nonumber\\
\end{eqnarray}
and therefore as $\mbox{SNR}\rightarrow\infty$, the innerbound on degrees of freedom $d(\alpha)\geq 1-\alpha$ which coincides with the outerbound for $\alpha\leq 1/2$.
\end{proof}
\section{Conclusion}
We found the generalized degrees of freedom per user for the $K$ user symmetric Gaussian interference channel to be independent of the number of users with the exception of a singularity at $\alpha=1$. The symmetric channel model itself imposes a special structure on the channel coefficients which plays an important role in the achievable schemes. This work reaffirms the importance of the deterministic channel model of \cite{Avestimehr_Diggavi_Tse} as well as the importance of structured coding for interference networks with more than $2$ users. For future work, many important directions remain open. For example, extensions to asymmetric channels, complex signals and channel coefficients, time varying and frequency selective channels with random coefficients, and multiple antennas are all interesting directions. Another promising direction is to use the generalized degrees of freedom as a guide to pursue capacity characterizations within a constant number of bits in the manner of \cite{Etkin_Tse_Wang} for more than $2$ users. 
\bibliographystyle{ieeetr}
\bibliography{Thesis}
\end{document}